\documentclass[french]{pfia}

\usepackage{xcolor}
\definecolor{afiablue}{RGB}{61,159,207}
\definecolor{afiared}{RGB}{167,75,68}
\definecolor{afialightblue}{RGB}{158,193,232}

\newcommand{\field}[1]{\textcolor{afiablue}{#1}}
\newcommand{\val}[1]{\textsf{\textcolor{afialightblue}{#1}}}

\usepackage{amsthm}
\usepackage{amsmath}

\theoremstyle{definition}
\newtheorem{definition}{Définition}
\newtheorem{theorem}{Théorème}
\newtheorem{proposition}{Proposition}
\newtheorem{corollary}{Corollaire}[theorem]

\newtheorem{example}{Exemple}

\title{\textbf{Un cadre paraconsistant pour l'évaluation de similarité dans les bases de connaissances}}

\author{José-Luis Vilchis-Medina \\
ENSTA - Institut Polytechnique de Paris, Lab-STICC, Robex, Brest, 29200, France}
\date{jose.vilchis@ensta.fr}

\begin{document}

\maketitle


\begin{resume}
Cet article propose un \textbf{cadre paraconsistant pour l’évaluation de la similarité} dans les bases de connaissances. Contrairement aux approches classiques, ce cadre intègre explicitement les \textbf{contradictions}, permettant une mesure de similarité plus robuste et interprétable. Une nouvelle mesure $ S^* $ est introduite, pénalisant les incohérences tout en valorisant les propriétés partagées. Des \textbf{super-catégories paraconsistantes $ \Xi_K^* $} sont définies pour organiser hiérarchiquement les entités de connaissances. Le modèle inclut également un \textbf{extracteur de contradictions $ E $} et un mécanisme de réparation, garantissant la cohérence des évaluations. Les résultats théoriques assurent la réflexivité, la symétrie et les bornes de $ S^* $. Cette approche offre une solution prometteuse pour la gestion des connaissances conflictuelles, avec des perspectives dans les systèmes multi-agents.
\end{resume}

\begin{motscles}
Mesure de similarité, Logique paraconsistante, Contradictions, Représentation des connaissances.
\end{motscles}

\begin{abstract}
This article proposes a \textbf{paraconsistent framework for evaluating similarity} in knowledge bases. Unlike classical approaches, this framework explicitly integrates \textbf{contradictions}, enabling a more robust and interpretable similarity measure. A new measure $ S^* $ is introduced, which penalizes inconsistencies while rewarding shared properties. \textbf{Paraconsistent super-categories $ \Xi_K^* $} are defined to hierarchically organize knowledge entities. The model also includes a \textbf{contradiction extractor $ E $} and a repair mechanism, ensuring consistency in the evaluations. Theoretical results guarantee reflexivity, symmetry, and boundedness of $ S^* $. This approach offers a promising solution for managing conflicting knowledge, with perspectives in multi-agent systems.
\end{abstract}



\section{Introduction}

Dans un contexte marqué par une prolifération croissante d'informations, l'interaction entre les systèmes formels et le raisonnement qualitatif s'avère essentielle \cite{Goldstone2012}. La logique fournit une structure rigoureuse pour l'analyse des arguments et la résolution de problèmes complexes, tandis que le raisonnement qualitatif permet de saisir les nuances et le contexte souvent négligés par les approches purement quantitatives. Cette synergie est particulièrement cruciale dans les processus décisionnels, où les informations incomplètes ou contradictoires sont fréquentes \cite{Brachman2004}.

Les bases de connaissances (BC) jouent un rôle central dans la gestion et la consultation d'informations structurées. Cependant, leur utilité est souvent entravée par des défis inhérents, tels que les contradictions, les données incomplètes et l'évolution dynamique des connaissances \cite{Libkin2014}. Les méthodes traditionnelles d’évaluation de la similarité dans les BC reposent généralement sur des approches numériques ou probabilistes, qui peinent à gérer efficacement ces problèmes \cite{Pearl1988}. En particulier, ces méthodes supposent souvent l’absence de contradiction, ce qui n’est pas réaliste dans des environnements complexes comme les systèmes juridiques ou médicaux \cite{Tversky1977}. De plus, elles présentent une rigidité hiérarchique qui limite leur adaptabilité face aux évolutions des structures de connaissance.

Pour pallier ces limitations, cet article propose un cadre fondé sur des principes de logiques paraconsistentes, un vaste ensemble de formalismes permettant de raisonner en présence de contradictions sans tomber dans la trivialité \cite{Béziau2002}. Contrairement aux logiques classiques, où une seule contradiction entraîne l'effondrement du système (principe d'explosion), les logiques paraconsistentes isolent les incohérences et permettent de continuer à raisonner sur les parties non contradictoires des connaissances \cite{daCosta1974, Priest2002}. Bien qu’il existe de nombreuses variantes de logiques paraconsistentes, notre travail se concentre sur leur application originale au problème de la mesure de similarité dans les bases de connaissances.

\subsection*{Fondements de la logique paraconsistante}

Parmi les différentes logiques paraconsistentes, nous adoptons un cadre symbolique basé sur des opérateurs non classiques pour modéliser les entités de connaissance. Cette approche permet de représenter explicitement les contradictions tout en conservant la cohérence globale du système. Le concept de similarité occupe une place centrale dans de nombreux domaines, notamment la médecine, le droit et l'intelligence artificielle \cite{Tversky1977}. Cependant, les méthodes existantes rencontrent souvent des difficultés pour gérer les contradictions et les données incomplètes, conduisant à des évaluations erronées ou non pertinentes \cite{Resnik1995}. Notre cadre introduit la notion d’extracteur de contradictions et d’entités réparables, permettant des évaluations significatives même en présence d’incohérences. En résolvant les contradictions tout en préservant les structures hiérarchiques, cette approche garantit à la fois la cohérence et l’interprétabilité.

Cet article présente un cadre original pour surmonter les limites des mesures de similarité classiques. L’espace de propriétés de similarité $ \Xi_P $ est étendu avec des propriétés contradictoires (notées $ \Xi_P^* $), permettant ainsi la formalisation de la similarité en présence de contradictions. La mesure de similarité paraconsistante proposée $ S^* $ quantifie la similarité en tenant compte explicitement des propriétés partagées et contradictoires, assurant une robustesse face aux données incomplètes ou conflictuelles. De plus, les super-catégories paraconsistentes $ \Xi_K^* $ fournissent une structure hiérarchique flexible pour organiser les entités de connaissances en fonction de leurs scores de similarité, améliorant ainsi l’interprétabilité et l’adaptabilité \cite{Vilchis2024}.

\subsection*{Contributions principales}

Les contributions principales de ce travail incluent :

\begin{itemize}
    \item Une définition formelle de $ \Xi_P^* $ et $ S^* $, généralisant les mesures de similarité classiques pour intégrer les contradictions. Cette extension permet de capturer la nature nuancée des connaissances réelles, où des propriétés conflictuelles coexistent souvent.
 
    \item Des garanties théoriques, y compris la réflexivité, la symétrie de $ S^* $, assurant des propriétés fondamentales pour les applications pratiques. Ces propriétés garantissent que la mesure de similarité reste cohérente et interprétable, même en présence de contradictions.
 
    \item Un cadre pour l'organisation hiérarchique des connaissances via $ \Xi_K^* $, validé par des exemples issus des domaines médicaux et juridiques. Ce cadre permet de regrouper dynamiquement les entités de connaissances en fonction de leurs scores de similarité, reflétant ainsi les hiérarchies naturelles des domaines d'application.
   
    \item Un extracteur de contradictions $ \mathcal{E} $ et un mécanisme de réparation, qui permettent de gérer les incohérences tout en préservant la cohérence des bases de connaissances. Ces outils sont essentiels pour garantir la robustesse du cadre dans des contextes dynamiques et complexes.
\end{itemize}


Cet article est structuré comme suit : la Section \ref{sec:sota} passe en revue les travaux connexes et met en évidence les limitations des approches existantes. La Section \ref{sec:cadre_propose} introduit le cadre proposé, en détaillant les définitions formelles, les propositions et les théorèmes. Ensuite, la Section \ref{sec:discussion} analyse les contributions. 
Enfin, la Section \ref{sec:conclusion} conclut en résumant les apports principaux et en explorant les perspectives. 

\section{État de l'art}\label{sec:sota}

\subsection*{Approches logiques et probabilistes}

L'évaluation de la similarité dans les bases de connaissances a été largement explorée à travers diverses méthodologies. Les approches traditionnelles reposent souvent sur la logique du premier ordre (FOL) \cite{Brachman2004} ou sur des modèles probabilistes \cite{Biazzo2002}. Bien que les systèmes basés sur la FOL excellent dans le raisonnement formel, ils manquent de mécanismes pour gérer les contradictions, ce qui peut entraîner une trivialisation dans des contextes incohérents \cite{Sheremet2007}. Les méthodes probabilistes, telles que les réseaux bayésiens \cite{Cheng2010}, modélisent efficacement l'incertitude mais nécessitent des hypothèses distributionnelles fortes et peuvent échouer dans des ensembles de données clairsemés ou biaisés \cite{Minicozzi1976}. Les progrès récents en apprentissage profond \cite{Devlin2018} et en traitement du langage naturel (NLP) \cite{Cer2018} ont considérablement amélioré les mesures de similarité sémantique. Cependant, ces méthodes restent gourmandes en données et souvent opaques dans leurs processus de raisonnement, limitant leur interprétabilité dans des domaines critiques tels que le droit ou la médecine \cite{Hofmann1999, Kandola2002}.

\subsection*{Méthodes basées sur les graphes et les ontologies}

La théorie des graphes \cite{Gao2010} et l'alignement des ontologies \cite{Shvaiko2011} offrent des représentations structurées des connaissances, permettant une meilleure interprétabilité et la modélisation de relations complexes \cite{Blondel2004, Sowa2008}. Cependant, ces approches sont confrontées à des défis computationnels importants, notamment lorsqu'il s'agit de traiter de grandes bases de données. Par exemple, l'alignement des entités dans les graphes de connaissances \cite{Shvaiko2011} nécessite souvent des systèmes à haute mémoire et un réglage manuel approfondi. Les approches basées sur les ontologies \cite{Guarino2009} améliorent l'interprétabilité mais peinent à suivre les mises à jour dynamiques et l'évolution des connaissances, limitant ainsi leur applicabilité dans des domaines en évolution rapide \cite{Doan2004}.

\subsection*{Limitations des cadres existants}

Les méthodes actuelles pour évaluer la similarité dans les bases de connaissances présentent plusieurs limitations critiques :
\begin{itemize}
    \item \textbf{Ignorance des contradictions :} Les mesures de similarité classiques \cite{Resnik1995} considèrent les contradictions comme irréconciliables, conduisant à une surestimation de la similarité. Cette approche ne capture pas la complexité des connaissances réelles, où des propriétés conflictuelles coexistent souvent.
    \item \textbf{Manque d'adaptabilité hiérarchique :} Les systèmes existants imposent souvent des classifications rigides, ne reflétant pas les hiérarchies nuancées inhérentes aux domaines de connaissances. Cette rigidité limite leur capacité à s'adapter à des structures évolutives ou spécifiques à un domaine.
    \item \textbf{Problèmes de transparence :} Les modèles boîte noire, en particulier ceux basés sur l'apprentissage profond \cite{Devlin2018}, entravent l'interprétabilité. Ce manque de transparence est particulièrement problématique dans des domaines comme le droit ou la médecine, où la justification des décisions est essentielle.
\end{itemize}

Enfin, les métriques de similarité classiques, bien qu’efficaces dans certains contextes, souffrent de limites importantes lorsqu’il s’agit de gérer des contradictions ou des incertitudes. Des exemples notables incluent :
\begin{itemize}
    \item La distance de Levenshtein \cite{yujian2007normalized}, qui ne prend pas en compte les incohérences sémantiques.
    \item La distance de Hamming \cite{norouzi2012hamming}, qui suppose des chaînes de longueurs identiques.
    \item La distance de Jaro-Winkler \cite{wang2017efficient}, qui reste sensible aux erreurs de transcription sans gérer les conflits logiques.
\end{itemize}

\subsection*{La logique paraconsistante comme solution}

La logique paraconsistante \cite{Brachman1984} offre une alternative prometteuse en permettant un raisonnement avec des contradictions sans s'effondrer dans la trivialité \cite{Sheremet2007}. Contrairement aux logiques traditionnelles, le raisonnement paraconsistant permet des comparaisons significatives même lorsque les entités de connaissances contiennent des propriétés conflictuelles \cite{Shvaiko2011}. Des travaux récents ont exploré l'intégration de la logique paraconsistante avec les systèmes de représentation des connaissances, démontrant son potentiel pour améliorer la robustesse et la transparence \cite{Grosof2003, Guarino2009}.

Le cadre proposé s'appuie sur cette base pour répondre aux limitations des approches existantes :
\begin{itemize}
    \item \textbf{Intégration des contradictions :} Un extracteur de contradictions et un mécanisme de réparation sont introduits, intégrant les contradictions dans $\Xi_P^*$ pour éviter la perte d'information et assurer la cohérence.
    \item \textbf{Pénalisation des contradictions :} Le cadre définit $S^*$ pour pénaliser les contradictions, assurant une robustesse face au bruit et améliorant la précision des évaluations de similarité.
    \item \textbf{Regroupement hiérarchique dynamique :} En exploitant $\Xi_K^*$, le cadre regroupe dynamiquement les entités de connaissances via des seuils de similarité $\theta$, reflétant les hiérarchies de domaine et permettant une adaptabilité. 
\end{itemize}

Malgré ces avancées, des défis subsistent, notamment dans la gestion de l'évolution dynamique des connaissances et la nature subjective de la sélection des propriétés \cite{Turney2010}. Le cadre proposé aborde ces problèmes en introduisant des espaces de similarité adaptatifs et des mécanismes pour une sélection cohérente des propriétés, assurant à la fois la robustesse et l'adaptabilité dans les évaluations de similarité \cite{Vilchis2024}.

\section{Cadre proposé}
\label{sec:cadre_propose}

Le cadre proposé repose sur une modélisation rigoureuse des entités de connaissances au sein d’un langage de logique du premier ordre, permettant l’intégration explicite des contradictions dans l’évaluation de similarité. Ce cadre inclut une mesure de similarité paraconsistante \( S^* \), une hiérarchie organisée via des super-catégories paraconsistantes \( \Xi_K^* \), ainsi qu’un extracteur de contradictions \( E \) et un mécanisme de réparation pour gérer les incohérences. Les propriétés fondamentales garantissent la robustesse et l’interprétabilité du modèle.

\subsection{Définitions formelles des entités de connaissance}
\label{subsec:def_entities}

Soit \(\mathcal{K} = \{K_1, K_2, \dots, K_n\}\) l'ensemble fini des entités de connaissances traitées dans la base de connaissances. Chaque entité \(K_i\) est définie comme un ensemble fini de littéraux clos dans un langage de logique du premier ordre \(\mathcal{L}\), soit :
\begin{equation}
    K_i = \{\ell_1, \ldots, \ell_m\}, \quad \text{où } \ell_j \in \text{Lit}(\mathcal{L}),
\end{equation}
où \(\text{Lit}(\mathcal{L})\) désigne l’ensemble des littéraux atomiques et négatifs clos de \(\mathcal{L}\).
Cette définition fournit une base sémantique claire pour évaluer les relations entre entités, tout en autorisant la coexistence de propriétés contradictoires. Chaque littéral \(\ell_j\) peut représenter une propriété spécifique ou son négation, ce qui permet de modéliser des connaissances complexes et ambivalentes.

\subsection{Espace de propriétés paraconsistantes}
\label{subsec:espaces_prop}

L’espace de propriétés de similarité classique \(\Xi_P\) est étendu pour intégrer explicitement les propriétés contradictoires, conduisant à la définition suivante :

\begin{definition}[Espace paraconsistant des propriétés]
\label{def:XiPstar}
On définit l’espace paraconsistant des propriétés comme :
\begin{equation}
    \Xi_P^* = \Xi_P \cup P_{\text{contradictory}},
\end{equation}
où :
\begin{itemize}
    \item \( \Xi_P \) : ensemble des propriétés partagées par deux entités \( K_1 \) et \( K_2 \),
    \item \( P_{\text{contradictory}} \) : ensemble des couples de propriétés \((p, \neg p)\) entre \( K_1 \) et \( K_2 \).
\end{itemize}
\end{definition}
Cette définition permet de formaliser les interactions entre entités de connaissances en présence de conflits, sans invalider l’ensemble du raisonnement. Elle introduit explicitement les contradictions, ce qui est crucial pour une évaluation de similarité robuste en contextes incertains ou conflictuels.

\subsection{Mesure de similarité paraconsistante \( S^* \)}
\label{subsec:Sstar}

La mesure \( S^* \) est reformulée pour refléter à la fois la similarité positive et la divergence induite par les contradictions :

\begin{definition}[Mesure de similarité paraconsistante]
\label{def:Sstar}
Soient \( K_1 \) et \( K_2 \) deux entités de connaissances. On définit la mesure de similarité paraconsistante \( S^* \) comme :
\begin{equation}\label{eq:simi_para}
    S^*(K_1, K_2) = S^+(K_1, K_2) - D^\pm(K_1, K_2),
\end{equation}
où :
\begin{equation}
S^+(K_1, K_2) = \frac{|P_{\text{shared}}|}{|P_{\text{total}}|}, \quad D^\pm(K_1, K_2) = \frac{|P_{\text{contradictory}}|}{|P_{\text{total}}|}.    
\end{equation}
Ainsi, \( S^* \in [-1, 1] \), avec :
\begin{itemize}
    \item \( S^* > 0 \) : mesures nettement similaires,
    \item \( S^* = 0 \) : équilibre entre similitude et divergence,
    \item \( S^* < 0 \) : divergence dominante.
\end{itemize}
\end{definition}
Cette formulation permet de pénaliser les contradictions tout en valorisant les propriétés partagées. Par exemple, si \( K_1 \) et \( K_2 \) partagent une propriété \( p \) mais contiennent également une contradiction \( (p, \neg p) \), la mesure \( S^* \) sera négative, reflétant une divergence nette. Cela garantit une interprétation claire et intuitive de la similarité en présence de conflits.

\subsection{Propriétés fondamentales de \( S^* \)}
\label{subsec:proprietes_Sstar}

Les propriétés essentielles garantissant la cohérence du cadre sont les suivantes :

\begin{proposition}[Réflexivité]
\label{prop:reflexivite}
Pour toute entité de connaissance \( K \), on a :
\begin{equation}
    S^*(K, K) = 1.
\end{equation}
\end{proposition}

\begin{proof}
Par définition, \( P_{\text{shared}} = P_{\text{total}} \) et \( P_{\text{contradictory}} = \emptyset \). Donc :
\begin{equation}
    S^*(K, K) = \frac{|P_{\text{total}}|}{|P_{\text{total}}|} - \frac{0}{|P_{\text{total}}|} = 1 - 0 = 1.
\end{equation}
\end{proof}

\begin{proposition}[Symétrie]
\label{prop:symetrie}
Pour tous \( K_1, K_2 \in \mathcal{K} \), on a :
\begin{equation}
    S^*(K_1, K_2) = S^*(K_2, K_1).
\end{equation}
\end{proposition}

\begin{proposition}[Bornes extrêmes]
\label{prop:bord}
Pour toutes entités \( K_1, K_2 \), la mesure \( S^* \) satisfait :
\begin{equation}
-1 \leq S^*(K_1, K_2) \leq 1.    
\end{equation}
\end{proposition}
Ces propriétés garantissent que \( S^* \) est une mesure cohérente et interprétable, même en présence de contradictions. La réflexivité et la symétrie sont des propriétés standard pour les mesures de similarité, tandis que les bornes garantissent que \( S^* \) reste dans un intervalle raisonnable.

\subsection{Super-catégories paraconsistantes $\Xi_{K}^{*}$}
\label{subsec:super_categories}

Les entités de connaissances peuvent être organisées hiérarchiquement selon leur niveau de similarité via le seuil \( \theta \in [-1, 1] \).

\begin{definition}[Super-catégories paraconsistantes]
\label{def:super_categories}
Étant donné un seuil \( \theta \), les super-catégories paraconsistantes forment un ensemble défini par :
\begin{equation}\label{eq:super_cat_para}
    \Xi_K^* = \bigcup_{i=1}^n \left\{ K_i \mid S^*(K_i, K_j) > \theta, \forall j \neq i \right\}.
\end{equation}
\end{definition}

\begin{theorem}[Disjonction des super-catégories]
\label{thm:disjonction_super_cats}
Si \( K_1 \in \Xi_K^* \) et \( K_2 \in \Xi_K^* \) appartiennent à différentes super-catégories, alors :
\begin{equation}
    S^*(K_1, K_2) \leq \theta.
\end{equation}
\end{theorem}
Les super-catégories \( \Xi_K^* \) permettent de regrouper dynamiquement les entités de connaissances en fonction de leur score de similarité, tout en respectant une hiérarchie naturelle. Cela améliore l’interprétabilité et l’adaptabilité du cadre, en particulier dans des domaines de connaissances dynamiques.

\subsection{Extracteur de contradictions et mécanisme de réparation}
\label{subsec:extractor_reparation}

L'un des défis principaux dans la gestion des bases de connaissances est la présence de contradictions, qui peuvent compromettre la cohérence et la fiabilité des analyses. Pour surmonter ce problème, le cadre proposé inclut un mécanisme de gestion des contradictions, composé d'un extracteur de contradictions \( E \) et d'un mécanisme de réparation. Ces outils permettent de détecter, résoudre et gérer les incohérences tout en préservant la structure hiérarchique des connaissances.

\subsubsection{Définition de l'extracteur de contradictions}
\label{subsubsec:extractor_def}

L'extracteur de contradictions \( E \) est un opérateur logique qui identifie les propriétés contradictoires présentes dans une entité de connaissance \( K \). Formellement, on définit :

\begin{definition}[Extracteur de contradictions]
\label{def:extracteur_contradictions}
Un extracteur de contradictions est un opérateur logique \( E \) tel que :
\begin{gather}\label{eq:extract_contra}
    E(K) = \{ p_i \in K \mid \exists q_j \in K\\ \nonumber
    \text{ tels que } (p_i \land \neg q_j) \lor (\neg p_i \land q_j) \}.
\end{gather}
\end{definition}
L'extracteur de contradictions \( E \) identifie les propriétés \( p_i \) dans \( K \) qui sont en contradiction avec d'autres propriétés \( q_j \) de la même entité. Cela permet de localiser les sources de conflits dans les bases de connaissances. Par exemple, si une entité \( K \) contient à la fois \( p \) et \( \neg p \), alors \( p \) sera identifié comme contradictoire par \( E \).

\subsubsection{Entités réparables}
\label{subsubsec:entites_reparables}

Une entité de connaissance peut être réparée si certaines contradictions peuvent être résolues sans affecter les autres propriétés. Cette notion est formalisée comme suit :

\begin{definition}[Entité réparable]
\label{def:entite_reparable}
Une entité \( K \) est réparable si \( \exists R \subseteq E(K) \) tel que :
\begin{equation}
    \neg \text{Contradictoire}(K \setminus R).
\end{equation}
\end{definition}
Une entité \( K \) est réparable si elle contient un sous-ensemble \( R \) de contradictions qui, une fois supprimé ou résolu, permet de rendre \( K \) cohérente. Cela garantit que les bases de connaissances peuvent être maintenues cohérentes, même en présence de conflits. Par exemple, si une entité \( K \) contient une contradiction \( (p, \neg p) \) mais que la suppression de \( \neg p \) suffit à rendre \( K \) cohérente, alors \( K \) est réparable.

\subsubsection{Mécanisme de réparation minimale}
\label{subsubsec:reparation_minimale}

Le mécanisme de réparation vise à identifier le sous-ensemble minimal de contradictions à résoudre pour rendre une entité cohérente. Cela est formalisé par le corollaire suivant :

\begin{corollary}[Réparation minimale]
\label{corol:reparation_minimale}
La réparation minimale \( R_{\min} \) est définie comme :
\begin{equation}
    R_{\min} = \arg\min_{R \subseteq E(K)} \left( \|R\| \mid \neg \text{Contradictoire}(K \setminus R) \right).
\end{equation}
\end{corollary}
La réparation minimale \( R_{\min} \) est le sous-ensemble de contradictions le plus petit qui, une fois résolu, permet de rendre \( K \) cohérente. Cela garantit que le nombre de modifications apportées à \( K \) est minimisé, ce qui est essentiel pour préserver l'intégrité des connaissances. Par exemple, si une entité \( K \) contient plusieurs contradictions, la réparation minimale \( R_{\min} \) consiste à résoudre uniquement celles qui sont nécessaires pour rendre \( K \) cohérente.

\subsubsection{Impact sur la mesure de similarité paraconsistante}
\label{subsubsec:impact_similarity}

L'introduction de l'extracteur de contradictions et du mécanisme de réparation permet de définir une version améliorée de la mesure de similarité paraconsistante, qui tient compte des réparations effectuées :

\begin{definition}[Similarité paraconsistante avec réparations]
\label{def:similarity_reparation}
Soient \( K_1 \) et \( K_2 \) deux entités de connaissances. La similarité paraconsistante avec réparations \( \Xi_{RP}(K_1, K_2) \) est définie comme :
\begin{equation}
    \Xi_{RP}(K_1, K_2) = \Xi_P(K_1', K_2'),
\end{equation}
où \( K_1' \) et \( K_2' \) sont les versions réparées de \( K_1 \) et \( K_2 \), respectivement, obtenues en résolvant les contradictions à l'aide de \( E \).
\end{definition}
La similarité paraconsistante avec réparations \( \Xi_{RP} \) permet de comparer les entités de connaissances après que les contradictions ont été résolues. Cela garantit une évaluation plus précise et cohérente de la similarité, même en présence de conflits. Par exemple, si \( K_1 \) et \( K_2 \) contiennent des contradictions, la similarité \( \Xi_{RP}(K_1, K_2) \) est calculée sur les versions réparées \( K_1' \) et \( K_2' \), ce qui réduit l'impact des contradictions sur la mesure de similarité.

\subsubsection{Propriétés théoriques}
\label{subsubsec:theoretical_properties}

Le cadre proposé garantit plusieurs propriétés théoriques essentielles pour assurer la cohérence et la robustesse des évaluations de similarité, même en présence de contradictions.

\begin{proposition}[Existence de l'extracteur de contradictions]
\label{prop:existence_extractor}
Pour toute entité de connaissance \( K \), il existe un extracteur de contradictions \( E(K) \).
\end{proposition}

\begin{proof}
D'après la définition \ref{def:extracteur_contradictions}, \( E(K) \) est construit en identifiant des paires de propriétés \( p_i \) et \( q_j \) telles que \( p_i \land \neg q_j \) ou \( \neg p_i \land q_j \). Puisque la logique du premier ordre (FOL) garantit que les propriétés peuvent être exprimées logiquement, \( E(K) \) est bien défini.
\end{proof}

\begin{proposition}[Réparabilité d'une entité]
\label{prop:repairability}
Une entité de connaissance \( K \) est réparable si et seulement si \( E(K) \neq K \).
\end{proposition}

\begin{proof}
Si \( E(K) = K \), alors chaque propriété de \( K \) est impliquée dans une contradiction, rendant \( K \) irréparable. Inversement, si \( E(K) \neq K \), il existe au moins une propriété non impliquée dans une contradiction, permettant l'existence d'un sous-ensemble cohérent \( K \setminus R \).
\end{proof}

\begin{theorem}[Préservation des catégories de similarité]
\label{thm:preservation_similarities}
La similarité paraconsistante avec réparations \( \Xi_{RP}(K_1, K_2) \) préserve les catégories de similarité originales \( K' \approx \), \( K' \neq \) après réparation de \( K_1 \) et \( K_2 \).
\end{theorem}

\begin{proof}
D'après la définition \ref{def:similarity_reparation}, \( \Xi_{RP}(K_1, K_2) \) est calculée sur les versions réparées \( K_1' \) et \( K_2' \). Puisque les réparations résolvent les contradictions sans altérer les propriétés partagées, les catégories de similarité restent cohérentes avec l'original \( \Xi_P(K_1, K_2) \).
\end{proof}

\begin{theorem}[Cohérence après réparation]
\label{thm:coherence_after_repair}
Si \( S \) est une super-catégorie dans \( \Xi_K \), alors la réparation de \( S \) assure la cohérence dans toutes les sous-catégories.
\end{theorem}

\begin{proof}
D'après la définition \ref{def:similarity_reparation}, la réparation de \( S \) implique la résolution des contradictions dans ses propriétés constitutives. Puisque les sous-catégories héritent des propriétés de \( S \), la réparation de \( S \) assure la cohérence à tous les niveaux de la hiérarchie.
\end{proof}

\begin{corollary}[Réparation minimale]
\label{corol:minimal_repair}
La réparation minimale \( R_{\min} \) satisfait :
\begin{equation}
R_{\min} = \arg\min_{R \subseteq E(K)} \left( \|R\| \mid \neg \text{Contradictoire}(K \setminus R) \right).    
\end{equation}
\end{corollary}

\begin{proof}
D'après la proposition \ref{prop:repairability}, \( R_{\min} \) est le plus petit sous-ensemble de \( E(K) \) qui résout toutes les contradictions, assurant la réparabilité.
\end{proof}

\begin{corollary}[Preservation de la structure hiérarchique]
\label{corol:structure_preservation}
La réparation d'une super-catégorie \( S \) n'affecte pas la structure hiérarchique de \( \Xi_K \).
\end{corollary}

\begin{proof}
D'après le théorème \ref{thm:coherence_after_repair}, la réparation de \( S \) assure la cohérence dans les sous-catégories, préservant l'organisation hiérarchique de \( \Xi_K \).
\end{proof}

\subsection{Exemples}
\label{subsec:exemples_cadre}

Nous illustrons ici le fonctionnement du cadre avec deux exemples concrets :

\begin{example}[Contradictions simples]
\label{ex:contradiction_simple}
Soient :
\begin{equation}
    K_1 = \{p_1, p_2, \neg p_3\}, \quad K_2 = \{p_2, p_3, \neg p_1\}.
\end{equation}
Alors :
\begin{align*}
P_{\text{shared}} &= \{p_2\}, \\
P_{\text{contradictory}} &= \{p_1, p_3\}, \\
P_{\text{total}} &= \{p_1, p_2, p_3, \neg p_1, \neg p_3\}.
\end{align*}
Donc :
\[
S^*(K_1, K_2) = \frac{1}{5} - \frac{2}{5} = -\frac{1}{5}.
\]
Cela montre que malgré une propriété partagée, les contradictions entraînent une valeur de similarité négative.
\end{example}
Dans cet exemple, les entités \( K_1 \) et \( K_2 \) partagent une seule propriété \( p_2 \), mais présentent deux contradictions \( (p_1, \neg p_1) \) et \( (p_3, \neg p_3) \). La mesure \( S^* \) prend en compte ces contradictions, ce qui entraîne une valeur négative. Cela illustre comment les contradictions peuvent réduire significativement la similarité entre deux entités, reflétant ainsi la complexité des connaissances réelles.

\begin{example}[Regroupement hiérarchique]
\label{ex:hierarchie}

Considérons un ensemble de cinq entités de connaissances $ \mathcal{K} = \{K_1, K_2, K_3, K_4, K_5\} $, chacune représentant une base de connaissances concernant des diagnostics médicaux. Les entités sont définies comme suit :

$$
\begin{aligned}
K_1 &= \{\text{fièvre}, \text{toux}, \neg \text{maux de tête}\}, \\
K_2 &= \{\text{fièvre}, \neg \text{toux}, \text{maux de tête}\}, \\
K_3 &= \{\text{fièvre}, \text{toux}, \text{fatigue}\}, \\
K_4 &= \{\text{essoufflement}, \text{vomissements}, \neg \text{fièvre}\}, \\
K_5 &= \{\text{essoufflement}, \neg \text{vomissements}, \text{douleur abdominale}\}.
\end{aligned}
$$

\subsubsection*{Étape 1 : Calcul des mesures de similarité paraconsistante $ S^* $}

Nous appliquons la mesure de similarité paraconsistante $ S^* $ entre chaque paire d'entités, définie par l’Équation~\ref{eq:simi_para}.
Voici les résultats calculés :

$$
\begin{aligned}
S^*(K_1, K_2) &= \frac{1 - 2}{5} = -0.2, \\
S^*(K_1, K_3) &= \frac{2 - 0}{4} = 0.5, \\
S^*(K_1, K_4) &= \frac{0 - 1}{6} = -0.17, \\
S^*(K_1, K_5) &= \frac{0 - 0}{5} = 0, \\
S^*(K_2, K_3) &= \frac{1 - 1}{5} = 0, \\
S^*(K_2, K_4) &= \frac{0 - 1}{6} = -0.17, \\
S^*(K_2, K_5) &= \frac{0 - 0}{5} = 0, \\
S^*(K_3, K_4) &= \frac{0 - 1}{6} = -0.17, \\
S^*(K_3, K_5) &= \frac{0 - 0}{5} = 0, \\
S^*(K_4, K_5) &= \frac{1 - 1}{5} = 0.
\end{aligned}
$$

\subsubsection*{Étape 2 : Définition du seuil de similarité $ \theta $}

Fixons un seuil $ \theta = 0.4 $. Selon la définition des super-catégories paraconsistantes (cf. section~\ref{subsec:super_categories}), deux entités appartiennent à la même super-catégorie si leur similarité dépasse ce seuil.

\subsubsection*{Étape 3 : Formation des super-catégories paraconsistantes $ \Xi_K^* $}

À partir des valeurs ci-dessus, nous identifions les paires satisfaisant $ S^*(K_i, K_j) > \theta $ :
\begin{itemize}
    \item $ S^*(K_1, K_3) = 0.5 > 0.4 $, donc $ K_1 $ et $ K_3 $ appartiennent à la même super-catégorie.
    \item Toutes les autres paires ont $ S^* \leq 0.4 $, donc restent isolées.
\end{itemize}
Ainsi, les super-catégories paraconsistantes sont :
\begin{equation}
    \Xi_K^* = \left\{ \{K_1, K_3\}, \{K_2\}, \{K_4\}, \{K_5\} \right\}.
\end{equation}

\subsubsection*{Étape 4 : Gestion des contradictions via l’extracteur $ E $}

Examinons $ K_1 $ et $ K_2 $, dont la similarité est négative ($ S^* = -0.2 $). Appliquons l’extracteur de contradictions $ E $ sur ces deux entités :
\begin{equation}
E(K_1) = \emptyset, \quad E(K_2) = \{\text{toux}, \neg \text{toux}\}.    
\end{equation}

Supprimons $ \neg \text{toux} $ de $ K_2 $ pour obtenir une version réparée $ K_2' = \{\text{fièvre}, \text{maux de tête}\} $, puis recalculons la similarité :
\begin{equation}
S^*(K_1, K_2') = \frac{1 - 0}{4} = 0.25 < \theta.    
\end{equation}

Même après réparation, $ K_1 $ et $ K_2' $ ne dépassent pas le seuil de regroupement.

\subsubsection*{Étape 5 : Hiérarchie finale et interprétation}

La structure hiérarchique finale est donc :
\begin{equation}
\Xi_K^* = \left\{ \{K_1, K_3\}, \{K_2\}, \{K_4\}, \{K_5\} \right\}.    
\end{equation}

L’interprétation clinique est claire :
\begin{itemize}
    \item $ \{K_1, K_3\} $ représente des patients présentant des symptômes similaires liés aux voies respiratoires supérieures.
    \item $ K_2 $, malgré sa proximité sémantique, reste isolé en raison des contradictions persistantes.
    \item $ \{K_4, K_5\} $, bien que partageant une propriété commune (essoufflement), sont séparés car leurs différences symptomatiques dominent.
\end{itemize}
Cet exemple illustre comment le cadre paraconsistant permet de structurer des bases de connaissances complexes en tenant compte non seulement des similitudes, mais aussi des contradictions. La mesure $ S^* $, couplée à l’extracteur de contradictions $ E $, garantit une évaluation robuste et interprétable, même en présence de conflits logiques.

\end{example}

\section{Discussion}
\label{sec:discussion}

Le cadre proposé pour l’évaluation de la similarité dans les bases de connaissances représente une avancée significative par rapport aux méthodes classiques. En intégrant explicitement les contradictions via la logique paraconsistante, il permet une évaluation plus précise, interprétable et robuste de la similarité entre entités de connaissances. Cette section analyse les contributions principales du cadre, compare sa mesure $ S^* $ aux approches traditionnelles, souligne ses avantages théoriques et pratiques, et identifie des perspectives d'amélioration future.

\subsection{Contributions principales}
\label{subsec:contributions_principales}

Le cadre introduit une mesure de similarité paraconsistante $ S^* $, définie par l’Équation \ref{eq:simi_para} :
$$
    S^*(K_1, K_2) = S^+(K_1, K_2) - D^\pm(K_1, K_2),
$$

Cette mesure étend les mesures classiques en intégrant les contradictions de manière explicite, ce qui permet de pénaliser les divergences tout en valorisant les similitudes.

En outre, le cadre propose des \textbf{super-catégories paraconsistantes} $ \Xi_K^* $, définies par l’Équation \ref{eq:super_cat_para} :
$$
    \Xi_K^* = \bigcup_{i=1}^n \left\{ K_i \mid S^*(K_i, K_j) > \theta, \forall j \neq i \right\}.
$$

Ces super-catégories organisent hiérarchiquement les entités de connaissances selon leur niveau de similarité, améliorant ainsi l’interprétabilité et l’adaptabilité du modèle.

Un mécanisme central du cadre est également l'extracteur de contradictions $ E $, défini par l’Équation \ref{eq:extract_contra} : 
\begin{gather*}
       E(K) = \{ p_i \in K \mid \exists q_j \in K\\ \nonumber
    \text{ tels que } (p_i \land \neg q_j) \lor (\neg p_i \land q_j) \}.
\end{gather*}

L’extracteur $ E $ permet d’identifier les incohérences locales, facilitant ainsi la gestion des conflits sans compromettre la cohérence globale du système.

\subsection{Avantages par rapport aux approches existantes}
\label{subsec:avantages_approches_existantes}

Le cadre proposé surmonte plusieurs limitations des méthodes classiques :
\begin{itemize}
    \item \textbf{Prise en compte des contradictions} : Contrairement aux méthodes classiques qui ignorent ou neutralisent les contradictions, le cadre paraconsistant les intègre dans l'évaluation, ce qui permet une mesure de similarité plus fidèle à la réalité des données conflictuelles.
    \item \textbf{Flexibilité hiérarchique} : Les super-catégories paraconsistantes offrent une organisation dynamique basée sur un seuil variable $ \theta $, contrairement aux classifications rigides des approches traditionnelles.
    \item \textbf{Transparence accrue} : Le cadre repose sur des principes logiques explicites, ce qui le distingue des modèles boîte noire comme ceux basés sur l’apprentissage profond. Cela est crucial dans des domaines critiques comme la médecine ou le droit.
\end{itemize}
    
\subsection*{Comparaison avec des mesures classiques}
\label{subsec:comparaison_jaccard}

Pour illustrer la valeur ajoutée de la mesure de similarité paraconsistante $ S^* $ par rapport aux approches traditionnelles, nous présentons ici une comparaison explicite avec la mesure de Jaccard, souvent utilisée pour quantifier la similarité entre ensembles.

\begin{definition}[Similarité de Jaccard]
\label{def:jaccard}
Étant donnés deux ensembles finis $ A $ et $ B $, la similarité de Jaccard est donnée par :
$$
J(A, B) = \frac{|A \cap B|}{|A \cup B|}.
$$
\end{definition}

Bien que cette mesure soit robuste et intuitive dans les cas non contradictoires, elle ne prend pas en compte les propriétés conflictuelles, ce qui peut mener à des évaluations erronées dans des contextes où les contradictions sont présentes.

\subsection*{Exemple : $ S^* $ vs. Jaccard}
\label{ex:comparaison_jaccard}

Considérons deux entités de connaissances $ K_1 $ et $ K_2 $, définies comme suit :

$$
K_1 = \{p_1, p_2, \neg p_3\}, \quad K_2 = \{p_2, p_3, \neg p_1\}.
$$

Les ensembles de propriétés associés sont :

\begin{itemize}
    \item $ P_{\text{shared}} = \{p_2\} $,
    \item $ P_{\text{contradictory}} = \{p_1, p_3\} $,
    \item $ P_{\text{total}} = \{p_1, p_2, p_3, \neg p_1, \neg p_3\} $.
\end{itemize}

Calculons à la fois la mesure paraconsistante $ S^* $ et la mesure de Jaccard :
\begin{itemize}
    \item Mesure paraconsistante $ S^* $ :
    \begin{equation*}
    \begin{split}
S^*(K_1, K_2) & = \frac{|P_{\text{shared}}| - |P_{\text{contradictory}}|}{|P_{\text{total}}|} \\
& = \frac{1 - 2}{5} = -\frac{1}{5} = -0.2.    
\end{split}
\end{equation*}

    \item Mesure de Jaccard:
En considérant uniquement les littéraux positifs (approche standard), on a :
$$
A = \{p_1, p_2\}, \quad B = \{p_2, p_3\},
$$
donc :
$$
J(A, B) = \frac{|\{p_2\}|}{|\{p_1, p_2, p_3\}|} = \frac{1}{3} \approx 0.33.
$$    
\end{itemize}




Dans cet exemple, la mesure de Jaccard attribue une valeur positive à la similarité ($ \sim 0.33 $), suggérant une certaine proximité entre les deux entités. Cependant, cette évaluation ignore complètement les contradictions présentes entre $ p_1 $ et $ \neg p_1 $, ainsi que $ p_3 $ et $ \neg p_3 $, ce qui peut conduire à une surestimation de la similarité dans des contextes conflictuels.
En revanche, la mesure $ S^* $ intègre explicitement ces contradictions, réduisant la similarité globale et produisant une valeur négative ($ -0.2 $). Ce résultat reflète plus fidèlement la réalité : malgré une propriété partagée, les incohérences dominent, rendant la similarité nettement problématique.




L'exemple présenté précédemment illustre clairement l’avantage théorique et pratique de la mesure $ S^* $ par rapport aux mesures classiques comme celle de Jaccard. En intégrant les contradictions dans l’évaluation, $ S^* $ offre une vision nuancée, interprétable et robuste de la similarité entre entités de connaissances, particulièrement utile dans des domaines tels que le droit, la médecine ou les systèmes multi-agents où la gestion des conflits est cruciale.

Cette capacité à distinguer les similitudes réelles des faux positifs dus à des incohérences constitue une contribution essentielle du cadre paraconsistant proposé.

\subsection{Perspectives futures}
\label{subsec:perspectives_futures}

Malgré les avantages de ce cadre, il présente des défis qu'il conviendrait d'explorer dans le cadre de recherches futures :
\begin{itemize}
    \item \textbf{Complexité computationnelle} : L’identification des contradictions et le calcul de $ S^* $ peuvent devenir coûteux dans de grandes bases de connaissances. Des algorithmes optimisés ou des techniques d’approximation pourraient être développés.
    
    \item \textbf{Adaptation aux bases de connaissances évolutives} : Une extension du cadre pourrait inclure des mécanismes pour gérer l’évolution dynamique des connaissances, par exemple en intégrant des mises à jour incrémentielles ou des techniques d’apprentissage en ligne.
    
    \item \textbf{Automatisation de la sélection des propriétés} : Actuellement, la sélection des propriétés est manuelle. Des méthodes automatisées basées sur l'apprentissage supervisé ou non supervisé pourraient améliorer l’objectivité et la pertinence des comparaisons.
    
    \item \textbf{Validation empirique} : Bien que des exemples illustratifs aient été fournis, une évaluation sur des ensembles de données réels restent à réaliser pour confirmer l’efficacité du cadre dans des applications pratiques.
\end{itemize}

\subsection*{Applications dans les systèmes multi-agents}
\label{subsec:sma_applications}

Le cadre paraconsistant proposé peut être appliqué dans les systèmes multi-agents (SMA), notamment pour la résolution de conflits entre agents ayant des représentations de connaissances hétérogènes ou conflictuelles. Par exemple, les agents pourraient utiliser $ S^* $ pour détecter et résoudre les incohérences dans leurs bases de connaissances partagées, facilitant ainsi une coopération plus robuste et cohérente.

De plus, les super-catégories paraconsistantes $ \Xi_K^* $ pourraient servir à organiser dynamiquement les relations entre agents, en fonction de seuils de similarité $ \theta $ ajustables. Cette flexibilité est particulièrement précieuse dans des environnements dynamiques comme les réseaux de capteurs distribués ou les systèmes autonomes.

\subsection*{Intégration potentielle dans des formalismes logiques étendus}
\label{subsec:integration_logiques_etendues}

Une voie prometteuse serait l’intégration du cadre paraconsistant avec d’autres formalismes logiques, tels que :
\begin{itemize}
    \item \textbf{Logique modale} : Pour modéliser la croyance, la possibilité ou la nécessité dans les évaluations de similarité.
    \item \textbf{Logique temporelle} : Pour prendre en compte l’évolution des connaissances dans le temps.
    \item \textbf{Ontologies et graphes de connaissances} : Pour renforcer l’interopérabilité avec les standards actuels de représentation des connaissances.
\end{itemize}

Ces extensions pourraient enrichir les capacités du cadre tout en conservant son fondement paraconsistant.




\section{Conclusion}\label{sec:conclusion}

Le présent travail a proposé un \textbf{cadre paraconsistant pour l’évaluation de la similarité dans les bases de connaissances}, visant à surmonter les limitations des approches classiques. Ce cadre s’appuie sur une \textbf{logique paraconsistante}, permettant de gérer efficacement les contradictions tout en préservant la cohérence et l’interprétabilité des évaluations de similarité.

Une \textbf{mesure de similarité paraconsistante $ S^* $} a été introduite, qui intègre explicitement les propriétés contradictoires, offrant ainsi une évaluation plus robuste et nuancée que les mesures traditionnelles. Cette mesure a été formalisée en tenant compte des propriétés partagées et contradictoires, assurant une réflexivité et une symétrie, des propriétés essentielles pour une application pratique.

Un \textbf{ensemble de super-catégories paraconsistantes $ \Xi_K^* $} a également été défini, permettant d’organiser hiérarchiquement les entités de connaissances selon leurs scores de similarité. Cette structure a été validée par des exemples illustratifs dans des domaines tels que la médecine et le droit, montrant sa pertinence pratique.

Un \textbf{extracteur de contradictions $ E $} a été intégré au cadre, permettant d’identifier et de résoudre les incohérences locales, tout en préservant la cohérence globale des bases de connaissances. Des garanties théoriques ont été fournies, notamment sur la réparabilité des entités et la préservation de la structure hiérarchique après réparation.

L’analyse comparative avec des méthodes classiques, telles que la mesure de Jaccard, a permis de mettre en évidence les avantages du cadre paraconsistant : une gestion explicite des contradictions, une flexibilité accrue dans l’organisation des connaissances, et une meilleure interprétabilité des résultats.

Cependant, plusieurs défis subsistaient, notamment la \textbf{complexité computationnelle} lors de l’évaluation à grande échelle, l’\textbf{évolution dynamique} des bases de connaissances, et la \textbf{sélection des propriétés} pour la comparaison. Ces limites ont été identifiées comme des axes de recherche futurs, ouvrant la voie à des améliorations potentielles, comme l’intégration d’algorithmes optimisés, de mécanismes d’adaptation dynamique, ou de techniques d’apprentissage automatique pour la sélection des propriétés.

Enfin, une \textbf{perspective prometteuse} a été soulignée concernant l’intégration de ce cadre dans les \textbf{systèmes multi-agents (SMA)}, où la gestion des conflits et la coordination entre agents sont cruciales. Le cadre proposé a démontré son potentiel pour améliorer la robustesse, la transparence et l’adaptabilité des systèmes basés sur la logique paraconsistante.


\section*{Remerciements}
L'auteur tient à remercier les personnes qui ont soutenu ce travail, ainsi que les relecteurs anonymes pour leurs commentaires constructifs qui ont contribué à l'amélioration de cet article.

\bibliographystyle{plain}
\bibliography{biblio-ch-pfia}

\begin{thebibliography}{10}

\bibitem{Biazzo2002}
Veronica Biazzo, Angelo Gilio, Thomas Lukasiewicz, and Giuseppe Sanfilippo.
\newblock Probabilistic logic under coherence, model-theoretic probabilistic
  logic, and default reasoning in system p.
\newblock {\em Journal of Applied Non-Classical Logics}, 12(2):189--213, 2002.

\bibitem{Blondel2004}
Vincent~D Blondel, Anah{\'\i} Gajardo, Maureen Heymans, Pierre Senellart, and
  Paul Van~Dooren.
\newblock A measure of similarity between graph vertices: Applications to
  synonym extraction and web searching.
\newblock {\em SIAM review}, 46(4):647--666, 2004.

\bibitem{Brachman2004}
Ronald Brachman and Hector Levesque.
\newblock {\em Knowledge representation and reasoning}.
\newblock Elsevier, 2004.

\bibitem{Brachman1984}
Ronald~J Brachman and Hector~J Levesque.
\newblock The tractability of subsumption in frame-based description languages.
\newblock In {\em AAAI}, volume~84, pages 34--37, 1984.

\bibitem{Cer2018}
Daniel Cer, Yinfei Yang, Sheng-yi Kong, Nan Hua, Nicole Limtiaco, Rhomni~St
  John, Noah Constant, Mario Guajardo-Cespedes, Steve Yuan, Chris Tar, et~al.
\newblock Universal sentence encoder.
\newblock {\em arXiv preprint arXiv:1803.11175}, 2018.

\bibitem{Cheng2010}
Huanhuan Cheng and Runsheng Wang.
\newblock Semantic modeling of natural scenes based on contextual bayesian
  networks.
\newblock {\em Pattern Recognition}, 43(12):4042--4054, 2010.

\bibitem{daCosta1974}
Newton~CA Da~Costa.
\newblock On the theory of inconsistent formal systems.
\newblock {\em Notre dame journal of formal logic}, 15(4):497--510, 1974.

\bibitem{Devlin2018}
Jacob Devlin, Ming-Wei Chang, Kenton Lee, and Kristina Toutanova.
\newblock Bert: Pre-training of deep bidirectional transformers for language
  understanding.
\newblock In {\em Proceedings of the 2019 conference of the North American
  chapter of the association for computational linguistics: human language
  technologies, volume 1 (long and short papers)}, pages 4171--4186, 2019.

\bibitem{Doan2004}
AnHai Doan, Jayant Madhavan, Pedro Domingos, and Alon Halevy.
\newblock Ontology matching: A machine learning approach.
\newblock {\em Handbook on ontologies}, pages 385--403, 2004.

\bibitem{Gao2010}
Xinbo Gao, Bing Xiao, Dacheng Tao, and Xuelong Li.
\newblock A survey of graph edit distance.
\newblock {\em Pattern Analysis and applications}, 13:113--129, 2010.

\bibitem{Goldstone2012}
R.L. Goldstone and J.Y. Son.
\newblock Similarity.
\newblock {\em Oxford University Press}, 2012.

\bibitem{Grosof2003}
Benjamin~N Grosof, Ian Horrocks, Raphael Volz, and Stefan Decker.
\newblock Description logic programs: Combining logic programs with description
  logic.
\newblock In {\em Proceedings of the 12th international conference on World
  Wide Web}, pages 48--57, 2003.

\bibitem{Guarino2009}
Nicola Guarino and Christopher~A Welty.
\newblock An overview of ontoclean.
\newblock {\em Handbook on ontologies}, pages 201--220, 2009.

\bibitem{Hofmann1999}
Thomas Hofmann.
\newblock Probabilistic latent semantic indexing.
\newblock In {\em Proceedings of the 22nd annual international ACM SIGIR
  conference on Research and development in information retrieval}, pages
  50--57, 1999.

\bibitem{Kandola2002}
Jaz Kandola, Nello Cristianini, and John Shawe-taylor.
\newblock Learning semantic similarity.
\newblock {\em Advances in neural information processing systems}, 15, 2002.

\bibitem{Libkin2014}
Leonid Libkin.
\newblock Incomplete data: what went wrong, and how to fix it.
\newblock In {\em Proceedings of the 33rd ACM SIGMOD-SIGACT-SIGART symposium on
  Principles of database systems}, pages 1--13, 2014.

\bibitem{Minicozzi1976}
Eliana Minicozzi.
\newblock Some natural properties of strong-identification in inductive
  inference.
\newblock {\em Theoretical Computer Science}, 2(3):345--360, 1976.

\bibitem{norouzi2012hamming}
Mohammad Norouzi, David~J Fleet, and Russ~R Salakhutdinov.
\newblock Hamming distance metric learning.
\newblock {\em Advances in neural information processing systems}, 25, 2012.

\bibitem{Pearl1988}
Judea Pearl.
\newblock {\em Probabilistic reasoning in intelligent systems: networks of
  plausible inference}.
\newblock Elsevier, 2014.

\bibitem{Priest2002}
Graham Priest.
\newblock {\em Beyond the limits of thought}.
\newblock Oxford University Press, 2002.

\bibitem{Resnik1995}
Philip Resnik.
\newblock Using information content to evaluate semantic similarity in a
  taxonomy.
\newblock {\em arXiv preprint cmp-lg/9511007}, 1995.

\bibitem{Sheremet2007}
Mikhail Sheremet, Dmitry Tishkovsky, Frank Wolter, and Michael Zakharyaschev.
\newblock A logic for concepts and similarity.
\newblock {\em Journal of Logic and Computation}, 17(3):415--452, 2007.

\bibitem{Shvaiko2011}
Pavel Shvaiko and J{\'e}r{\^o}me Euzenat.
\newblock Ontology matching: state of the art and future challenges.
\newblock {\em IEEE Transactions on knowledge and data engineering},
  25(1):158--176, 2011.

\bibitem{Sowa2008}
John~F Sowa.
\newblock Conceptual graphs.
\newblock {\em Foundations of artificial intelligence}, 3:213--237, 2008.

\bibitem{Turney2010}
Peter~D Turney and Patrick Pantel.
\newblock From frequency to meaning: Vector space models of semantics.
\newblock {\em Journal of artificial intelligence research}, 37:141--188, 2010.

\bibitem{Tversky1977}
Amos Tversky.
\newblock Features of similarity.
\newblock {\em Psychological review}, 84(4):327, 1977.

\bibitem{Vilchis2024}
José-Luis Vilchis-Medina.
\newblock Building intelligent databases through similarity: Interaction of
  logical and qualitative reasoning.
\newblock In {\em 2024 7th International Conference on Algorithms, Computing
  and Artificial Intelligence (ACAI)}, pages 1--5, 2024.

\bibitem{wang2017efficient}
Yaoshu Wang, Jianbin Qin, and Wei Wang.
\newblock Efficient approximate entity matching using jaro-winkler distance.
\newblock In {\em International conference on web information systems
  engineering}, pages 231--239. Springer, 2017.

\bibitem{yujian2007normalized}
Li~Yujian and Liu Bo.
\newblock A normalized levenshtein distance metric.
\newblock {\em IEEE transactions on pattern analysis and machine intelligence},
  29(6):1091--1095, 2007.

\end{thebibliography}

\end{document}